\documentclass[11pt]{elsarticle}%
\usepackage{algorithm}
\usepackage{algpseudocode}
\usepackage{amssymb}
\usepackage{amsfonts}
\usepackage{amsmath}
\usepackage{graphicx}%
\usepackage[caption=false]{subfig}
\providecommand{\U}[1]{\protect\rule{.1in}{.1in}}
\usepackage{tikz}
\usetikzlibrary{shapes}
\usepackage{pgfplots}
\usepackage{verbatim}
\usepgfplotslibrary{statistics}
\pgfplotsset{compat=1.8}
\usepackage{amsmath,amssymb,amsthm}
\usepackage{tikz}
\usetikzlibrary{positioning}
\usepackage[numbered]{matlab-prettifier}

\usepackage{url}
\usepackage{hyperref}

\definecolor{dgreen}{rgb}{0.078,0.418,0.184}

\usepackage{graphicx}
\usepackage{array}
\usepackage{booktabs}

\newtheorem{theorem}{Theorem}

\newtheorem{lemma}[theorem]{Lemma}

\newtheorem{problem}[theorem]{Problem}

\newtheorem*{construction}{Construction}

\renewcommand{\P}{\textnormal{Pr}}

\renewcommand{\O}{\mathcal{O}}

\begin{document}
	
\title{Minimizing the effective graph resistance by adding links is NP-hard}

\author[1,2]{Robert E. Kooij\corref{cor1}}
\ead{R.E.Kooij@tudelft.nl}
\author[1]{Massimo A. Achterberg}
\ead{M.A.Achterberg@tudelft.nl}
\cortext[cor1]{Corresponding author}
\affiliation[1]{organization={Faculty of Electrical Engineering, Mathematics and Computer Science, Delft University of Technology},
addressline={\\P.O. Box 5031},
postcode={2600 GA},
city={Delft},
country={The Netherlands}}
\affiliation[2]{organization={Unit ICT, Strategy \& Policy, TNO (Netherlands Organisation for Applied Scientific Research)},
addressline={\\P.O. Box 96800},
postcode={2509 JE},
city={The Hague},
country={The Netherlands}}

\begin{abstract}
    The effective graph resistance, also known as the Kirchhoff index, is metric that is used to quantify the robustness of a network. We show that the optimisation problem of minimizing the effective graph resistance of a graph by adding a fixed number of links, is NP-hard.
\end{abstract}

\begin{keyword}
effective graph resistance \sep graph augmentation
\end{keyword}

\maketitle

\section{Introduction}
Many network metrics have been utilised to quantify the robustness of a network, see for instance \cite{cats},~\cite{hale},~\cite{yakup},~\cite{jose},   \cite{schneider}. Freitas \textit{et al.}~\cite{freitas} classify robustness metrics into three types: 
metrics based on structural properties, such as edge connectivity or diameter; metrics based on the spectrum of the adjacency matrix, such as the spectral radius or spectral gap; and metrics based on the spectrum of the Laplacian matrix, for instance the algebraic connectivity and the effective graph resistance. In this paper we consider the following optimization problem: how to augment a given graph $G$ by adding at most $k$ links, such that the robustness metric of the augmented network is optimal.
As robustness metric we consider the effective graph resistance $R_G$, also known as the Kirchhoff index, see Ellens \textit{et al.}~\cite{ellens2011graphresistance}. The effective graph resistance not only covers the shortest path between any pair of nodes, but incorporates all paths between any two nodes. Because in addition $R_G$ decreases upon the addition of an edge to the graph \cite{ghosh2008graphresistance}, this makes the effective graph resistance a good metric to evaluate the robustness of a network. 
\newline
Predari \textit{et al.} refer to the optimization problem at hand as $k$-Graph Robustness Improvement Problem (k-GRIP) \cite{predari2020graphresistance},  in which one has to decide where $k$ links are to be added to a given network $G$, such that the robustness metric is optimised. 
Several researchers investigated $k$-GRIP for specific robustness metrics. For instance, Wang \textit{et al.} \cite{wang2008algebraic} considered $1$-GRIP, with as robustness metric the second-smallest eigenvalue of the Laplacian matrix, which was coined the algebraic connectivity by Fiedler~\cite{fiedler}. They suggest several strategies to decide which single link to add to the network, in order to increase the algebraic connectivity as much as possible. A nice overview of $k$-GRIP for the algebraic connectivity is presented in \cite{li2018algebraicconnectivity}. The NP-hardness of $k$-GRIP for the algebraic connectivity was proved in \cite{AC_NP}.
\newline
\newline
For the effective graph resistance, 1-GRIP was considered by Wang \textit{et al.}~
\cite{wang2014graphresistance}. They investigated different strategies, based upon topological and spectral properties of the graph, to determine the most optimal link to add, and derived a lower bound for $R_G$ after adding a single link.  Pizzuti \textit{et al.}~\cite{pizzuti2018graphresistance}, ~\cite{Pizzuti2023} proposed and evaluated several genetic algorithms to find the most optimal edge to add, in order to minimize $R_G$. Clemente \textit{et al.}~\cite{clemente2020graphresistance} studied $k$-GRIP for the effective graph resistance and gave lower bounds for $R_G$ upon the addition of $k$ links, under some mild conditions for $k$. For $k = 1$ the lower bound in \cite{clemente2020graphresistance} clearly outperforms the lower bound in \cite{wang2014graphresistance}. Predari \textit{et al.} \cite{predari2020graphresistance} also consider $k$-GRIP for the effective graph resistance. They focus on heuristics for $k$-GRIP based upon sampling and a fast approximation method to compute the effective graph resistance.
\newline
Although for some choices of the robustness metric, $k$-GRIP is known to be NP-hard, to the best of our knowledge this has not been proved yet for the effective graph resistance. The aim of this paper is to prove that augmenting a given graph $G$ by adding $k$ links, in order to minimize the effective graph resistance, is NP-hard.

\section{Definitions and main result}
In this paper we consider undirected, connected simple graphs $G=(V,E)$ without self-loops. Here $V$ denotes the set of $N$ vertices, while $E$ is the set of $L$ links connecting vertex pairs of $V$. The notation $i \sim j$ indication that nodes $i$ and $j$ are adjacent in $G$. We let $G^c = (V,E^c)$ denote the complementary graph of $G$, where $E^c = \{(u,v) | u, v \in V, u \neq v, (u,v) \not\in E\}$. 
The adjacency matrix $A$ of $G$ is an $N \times N$ symmetric matrix with elements $a_{ij}$ that are either 1 or 0 depending on whether there is a
link between nodes $i$ and $j$ or not. The Laplacian matrix $Q$ of $G$ is an $N \times N$ symmetric matrix $Q=\Delta - A$, where $\Delta = diag(d_i)$ is the $N \times N$ diagonal degree matrix with the elements $d_i = \sum_{j=1}^N a_{ij}$.  The eigenvalues of $Q$ are all real and non-negative and can be ordered as $0=\lambda_1 \leq \lambda_2 \leq \cdots \leq \lambda_N$.

Interpreting the graph $G$ as an electrical network whose edges are resistors of $1\Omega$, the effective resistance~$\omega_{ij}$ between node $i$ and $j$ can be computed based on Kirchoff's circuit laws. Then the \emph{effective graph resistance}~$R_G$, also known as the \emph{Kirchhoff index}, is defined as the sum of the resistances over all node pairs \cite{kleinrandic1993resistance}
\begin{equation}\label{eq_RG2}
    R_G(G) = \sum_{1 \leq i < j \leq N} \omega_{ij}.
\end{equation}
Ellens \textit{et al.}~\cite{ellens2011graphresistance} showed that the effective graph resistance can also be computed using the Laplacian eigenvalues $\lambda_k$ of the graph $G$ as
\begin{equation}\label{eq_RG}
    R_G(G) = N\sum_{k=2}^N\frac{1}{\lambda_k}.
\end{equation}
Ellens \textit{et al.}~\cite{ellens2011graphresistance} argued that the effective graph resistance is an appropriate robustness metric. Note that the smaller the value of $R_G$ the larger the robustness of the network. The smallest value of the effective graph resistance for a graph on $N$ nodes is obtained for the complete graph $K_N$ and satisfies $R_G(K_N) = N - 1$. 
We will show in this paper that adding a specified number of edges to a given graph, in order to minimize the effective graph resistance, is NP-hard. We will now give an explicit description of the considered optimization problem.

\begin{problem}[Minimum effective graph resistance augmentation problem]
    Given an undirected, connected, simple graph $G = (V,E)$, a non-negative integer $k$ and a non-negative threshold $t$, is there a subset $B \subseteq E^c$ of size $|B| \leq k$ such that the graph $H = (V, E \cup B)$ satisfies $R_G(H) \leq t$?
    \label{problem_RG}
\end{problem}

Problem~\ref{problem_RG} is clearly in NP, because given a graph $G$ and the set of added edges $B$, the correctness of the given solution can be verified by computing the eigenvalues of the Laplacian matrix, which is an $\O(N^3)$ operation. Then simply computing \eqref{eq_RG} and comparing the outcome with the given threshold~$t$ verifies the solution. Thus the minimum effective graph resistance augmentation problem is in NP.

Problem~\ref{problem_RG} is the decision version of the following optimisation problem: Given an undirected, connected, simple graph $G=(V,E)$ and a non-negative threshold~$t$, find a set of currently non-existent edges of minimum size to add to $G$ such that the effective graph resistance $R_G$ of the augmented graph is at most $t$. We prove in this work that Problem~\ref{problem_RG} is NP-hard, which immediately implies that the corresponding optimisation problem is also NP-hard. Thus, the problem of adding a specified number of edges to a graph to minimize the effective graph resistance is also NP-hard. We now state the main result of the paper.

\begin{theorem}
The minimum effective graph resistance augmentation
problem is NP-hard.
\label{T1}
\end{theorem}

\section{Proof of Theorem \ref{T1}}
The proof of Theorem \ref{T1} heavily relies on the proof of the NP-hardness of the maximum algebraic connectivity augmentation problem, as given in \cite{AC_NP}. The proof is by reduction of our augmentation problem to a problem for which NP-hardness has been proved, namely the 3-colorability problem, see \cite{GAREY1976237}.
For our proof we will use a construction and a lemma from \cite{AC_NP} and two additional lemma's.

\begin{construction}~\cite{AC_NP}
Given a graph $G=(V,E)$ with $n > 1$ vertices and $m$ edges, a graph $G' = (V',E')$ is constructed which consists of three disjoint copies $G_0, G_1$ and $G_2$ of $G$. This implies that each vertex $v \in V$ is copied to a vertex $v_i \in G_i$ and each edge $(u,v) \in E$ is copied to $(u_i,v_i) \in G_i$, for $i=1,2,3$. By construction the graph $G'$ has $3n$ vertices and $3m$ edges. We now consider the minimum effective graph resistance augmentation problem on $G'$ with $k = 3n^2 - 3m$, such that the augmented graph $H$ has at  most $3n^2$ edges and $t = \frac{9n-5}{2}$.
\end{construction}

Now, let $K_{n,n,n}$ denote the complete tripartite graph. In order to prove that the minimum effective graph resistance augmentation problem can be reduced to the 3-colorability problem, we will use the following three lemmas.

\begin{lemma}~\cite{AC_NP}
There exists a subset $B \subseteq (E')^c$ of size $|B| \leq k$ such that $H = (V',E' \cup B)$ is (isomorphic to) $K_{n,n,n}$ if and only if $G$ is 3-colorable.
\label{part1}
\end{lemma}

\begin{lemma}~\cite{Milovanovic} Let G be a simple connected graph with $N \geq 2$ vertices and $L$ edges. Then
\begin{equation}
R_G(G) \geq \frac{N^2(N-1)}{2L} - 1,
\end{equation}
with equality if and only if $G \cong K_N$, or $G \cong K_{N/2,N/2}$, or $G \in \Gamma_d$.
\label{Milovanovic}
\end{lemma}
Here, $\Gamma_d$ denotes a special class of $d$-regular graphs defined in \cite{Palacios}. 
Let $M(i)$ be the set of all neighbours of the vertex $i$, that is, $M(i) = \{k | k \in  V, k \sim i\}$, where $V$ denotes the set of vertices of the graph. Then for every $1 \leq  d \leq n-1$ the set $\Gamma_d$ denotes the set of all $d$-regular graphs with diameter~2 and satisfying $|M(i) \cap M(j)| = d$ for every pair of vertices $i, j$ that are not adjacent, i.e. $i \not\sim j$.

\begin{lemma}
    A graph $H = (V,E)$ with $N=3n$ vertices and $L \leq 3n^2$ edges for $n > 1$ satisfies $R_G(H) \leq \frac{9n-5}{2}$ if and only if $H$ is (isomorphic to) $K_{n,n,n}$.
    \label{main}
\end{lemma}

\begin{proof}
    First, we compute the effective graph resistance $R_G$ of the complete tripartite graph $K_{n,n,n}$ using Eq.\ \eqref{eq_RG2}. Gervacio \cite{gervacio2016graphresistance} derived the effective resistance between nodes in complete multipartite graphs as:
    \begin{align*}
        \omega_{ij} &= \frac{2}{N-m_i}, \qquad &\text{if } i,j \text{ are in the same partition} \\
        \omega_{ij} &= \frac{(N-1)(2N-m_i-m_j)}{N(N-m_i)(N-m_j)}, \qquad &\text{otherwise}
    \end{align*}
    where $m_i$ and $m_j$ represent the size of the partition of node $i$ and $j$ respectively. In our case, $N=3n$ and $m_i=m_j=n$. The number of node pairs in the same partition equals $3n(n-1)/2$, such that the number of pairs outside of the same partition equals $3n^2$. Then the effective graph resistance of the complete tripartite graph exactly equals $R_G(K_{n,n,n})=\frac{9n-5}{2}$.
\newline
\newline
Next, using $N=3n$ and $L \leq 3n^2$, according to Lemma \ref{Milovanovic} it follows $R_G(H) \geq \frac{9n^2(3n-1)}{6n^2}-1=\frac{9n-5}{2}$. By the condition $R_G(H) \leq \frac{9n-5}{2}$, we deduce that $R_G(H) = \frac{9n-5}{2}$. Also it follows that $L = 3n^2$ because $N=3n$ and $L < 3n^2$ would imply $R_G(H) > \frac{9n-5}{2}$ according to Lemma~\ref{Milovanovic}. Therefore the average degree of $H$ equals $2n$. Since $R_G(H)$ is equal to the lower bound given in Lemma~\ref{Milovanovic}, $H$ is either the complete graph $K_{3n}$, the complete bipartite graph $K_{3n/2,3n/2}$ or it is a $2n$-regular graph belonging to the class $\Gamma_{2n}$. First, assume $H \cong K_{3n}$. However, the number of links of $K_{3n}$ equals $\frac{3n(3n-1)}{2}$ which, for $n > 1$, is larger than $3n^2$, the number of links of $H$. Therefore 
$H \not\cong K_{3n}$. Next assume $H \cong K_{3n/2,3n/2}$, which can only hold for $n$ even. Then the number of links of $K_{3n/2,3n/2}$ equals $\frac{9n^2}{4}$
which is always smaller than $3n^2$, the number of links of $H$. Therefore 
$H \not\cong K_{3n/2,3n/2}$. Hence we conclude that the graph $H$ is $2n$-regular and belongs to the class $\Gamma_{2n}$.
\newline
\newline
We will now show that $H$ is isomorphic to $K_{n,n,n}$. We start with an arbitrary node of $H$ and label it as node 1. Because $H$ is $2n$-regular, node 1 has exactly $2n$ neighbours, see Fig.\ \ref{fig_NP1}. 

\begin{figure}[!ht]
\centering
    \subfloat[\label{fig_NP1}]{%
    \includegraphics[width=0.18\textwidth]{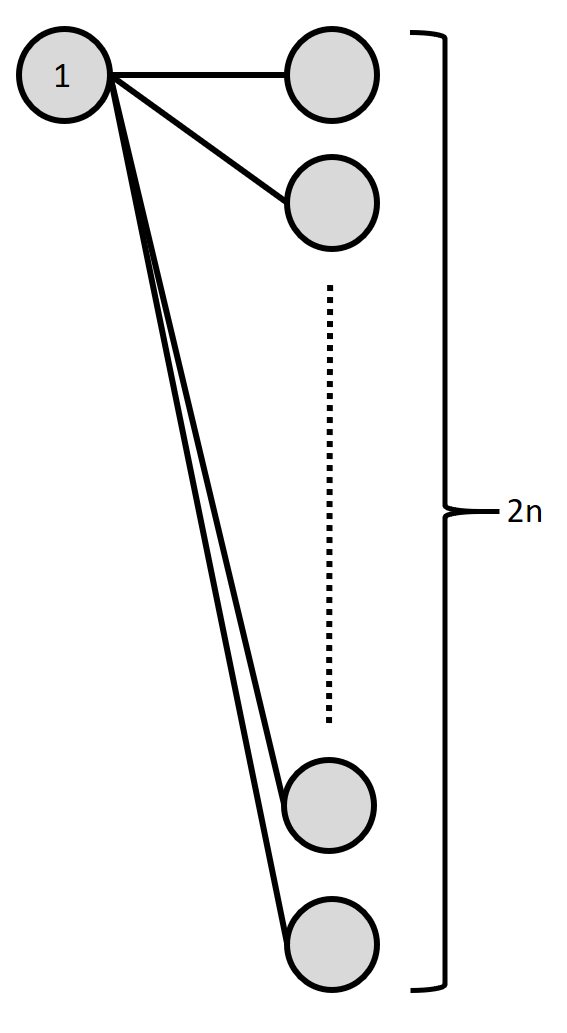}
    }\hspace{5mm}
    \subfloat[\label{fig_NP2}]{%
    \includegraphics[width=0.18\textwidth]{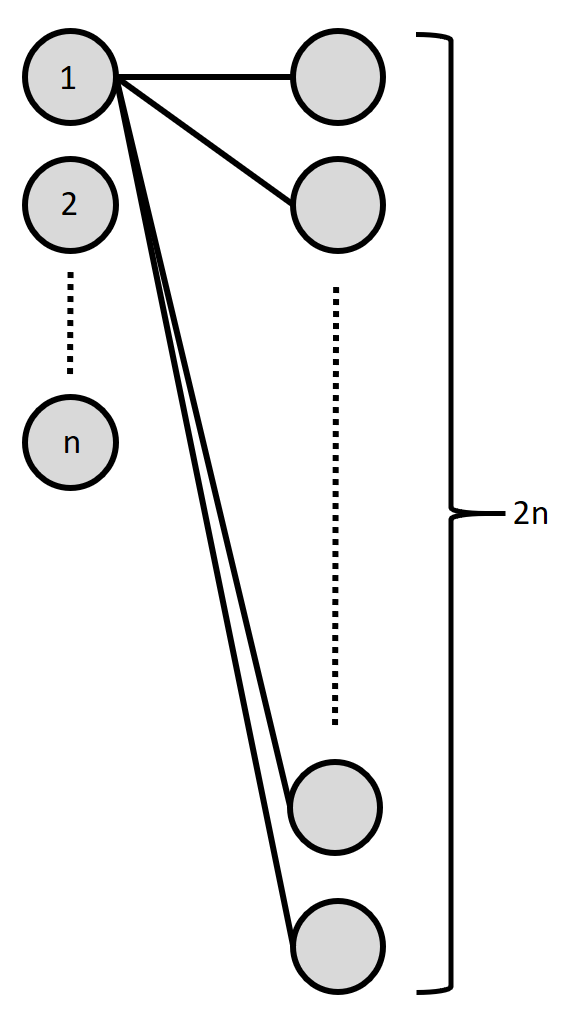}
    }\hspace{5mm}
    \subfloat[\label{fig_NP3}]{%
    \includegraphics[width=0.18\textwidth]{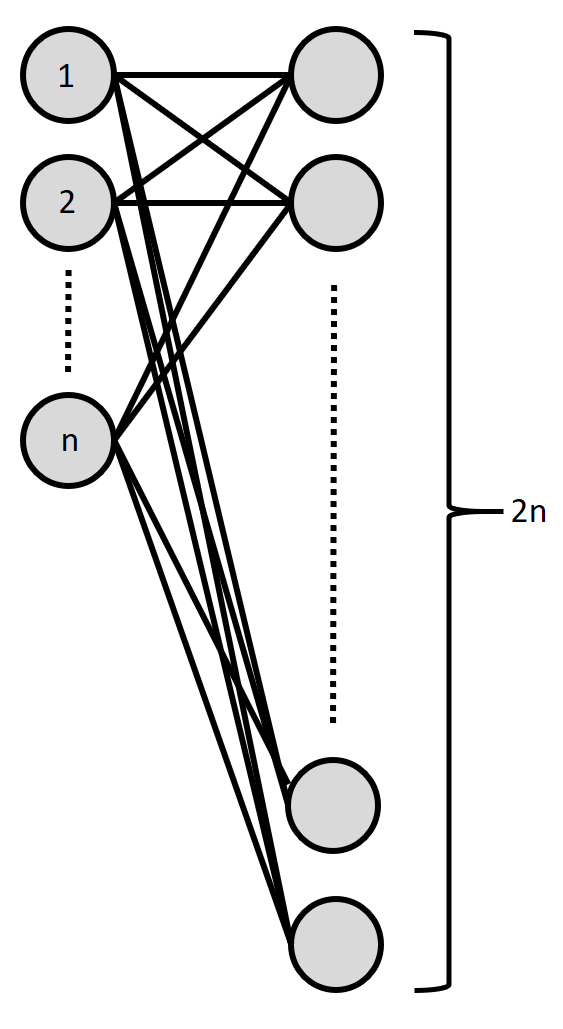}
    }
\caption{(a) Node 1 and its $2n$ neighbours. (b) Node 1, its $2n$ neighbours and the $n-1$ remaining nodes. (c) Nodes $\{1,\cdots,n\}$ and their connections to the other $2n$ nodes.}
\label{fig_NP123}
\end{figure}

The remaining $n-1$ nodes, other than node 1 and its $2n$ neighbours, cannot be adjacent to node~1 because it already has degree $2n$, by construction. We now label these nodes as nodes 2 until $n$, see Fig.\ \ref{fig_NP2}. Now, because $H$ belongs to the class $\Gamma_{2n}$ and nodes 2 until $n$ are not adjacent to node 1, each of the nodes 2 until $n$ has exactly the same neighbours as node 1, see Fig.\ \ref{fig_NP3}. 

Next, take an arbitrary node outside the set $\{1,2,\cdots,n\}$ and label it as $n+1$. To obtain degree $2n$, node $n+1$ needs to be adjacent to $n$ nodes outside the nodes $\{1,2,\cdots,n\}$. We label this set of $n$ adjacent nodes as $\{2n+1,\cdots,3n\}$, see Fig.\ \ref{fig_NP4}.

Finally, every node not in $\{1,2,\cdots,n+1\} \cup \{2n+1,\cdots,3n\}$ needs to share with node $n+1$ its neighbours $\{2n+1,\cdots,3n\}$, see Fig.\ \ref{fig_NP5}.

\begin{figure}[!ht]
\centering
    \subfloat[\label{fig_NP4}]{%
    \includegraphics[width=0.2\textwidth]{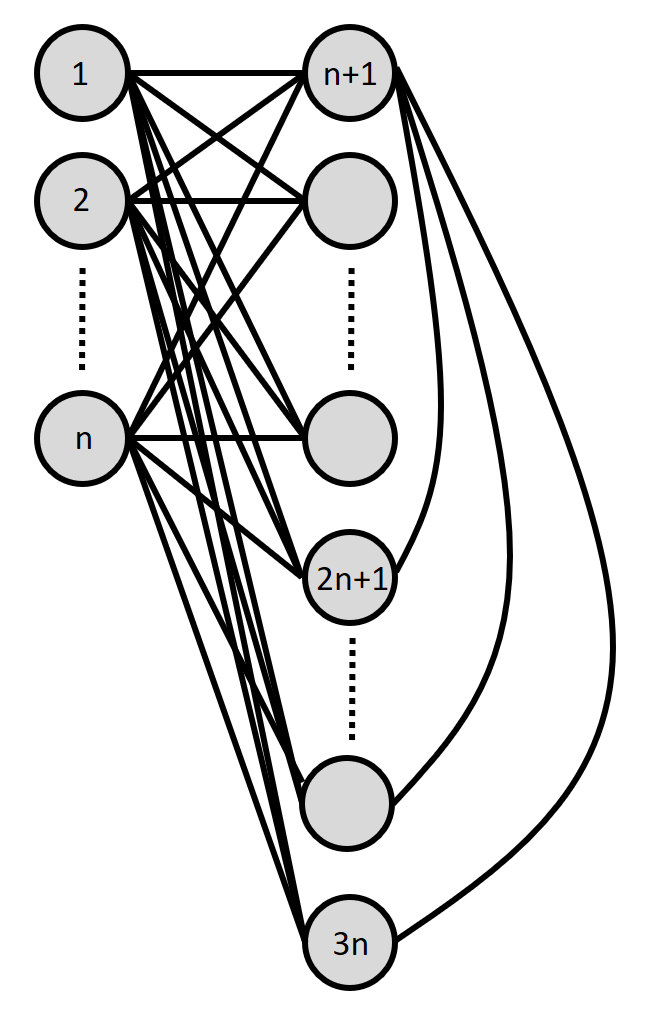}
    }\hspace{5mm}
    \subfloat[\label{fig_NP5}]{%
    \includegraphics[width=0.2\textwidth]{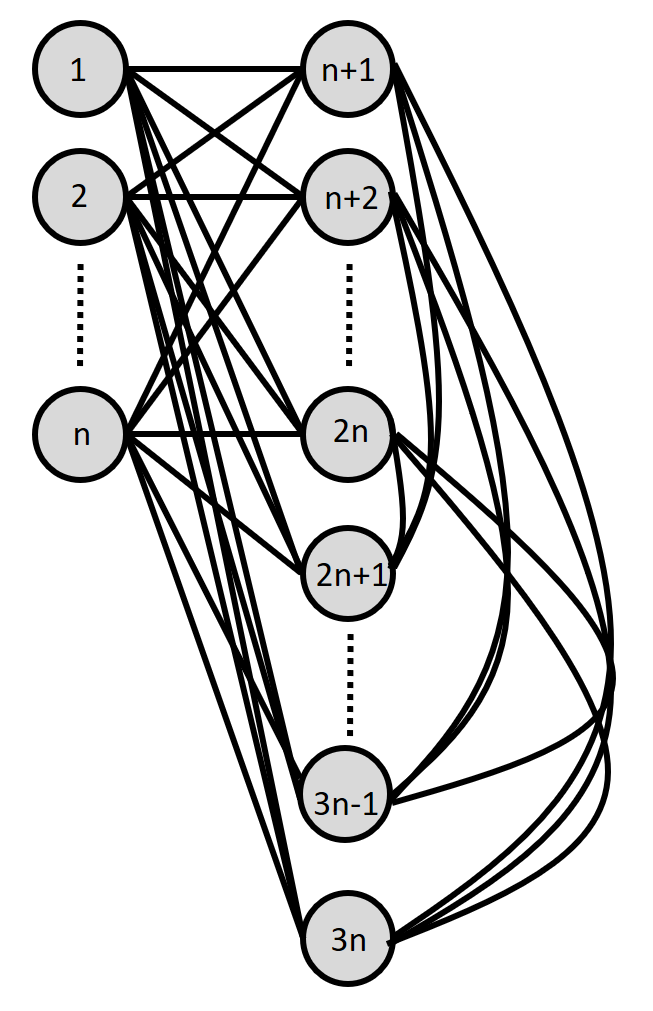}
    }
\caption{(a) Nodes $\{1,\cdots,n\}$, their connections to the other $2n$ nodes and the additional $n$ connections of node $n+1$. (b) All connections in graph $H$.}
\label{fig_NP45}
\end{figure}

Denote by $S_i$ the nodes labelled as $\{n(i-1)+1,n(i-1)+2,\cdots,n(i-1)+n\}$, for $i=1,2,3$. Then $|S_i| = n$, every node pair within $S_i$ is not adjacent and for every $i \neq j$ all nodes in $S_i$ are adjacent to all nodes in $S_j$. This proves that $H$ is a complete tripartite graph $K_{n,n,n}$. 
\end{proof}
Finally, Theorem \ref{T1} follows from combining Lemma \ref{part1} and \ref{main}.

{\footnotesize
    \bibliographystyle{apalike}

}

\section*{Statements \& Declarations}
\noindent
\textbf{Funding}
\newline
\textit{The authors declare that no funds, grants, or other support were received during the preparation of this manuscript}.
\newline
\newline
\textbf{Competing interests}
\newline
\textit{The authors have no relevant financial or non-financial interests to disclose}.
\newline
\newline
\textbf{Data availability}
\newline
\textit{No data was used for the research described in the article.}


\end{document}